\documentclass[journal,compsoc]{IEEEtran}
\usepackage{graphicx}
\usepackage{hhline}
\usepackage[table, dvipsnames]{xcolor}
\usepackage{caption}
\usepackage{subfig}
\usepackage{listings}
\usepackage{multirow}
\usepackage{enumitem}
\usepackage{tcolorbox}
\usepackage{booktabs}
\usepackage{inconsolata}
\usepackage{algorithm,algpseudocode}
\usepackage{mathtools}
\usepackage{amsthm}
\usepackage{ulem}
\usepackage{url}
\usepackage[
	n,
	operators,
	advantage,
	sets,
	adversary,
	landau,
	probability,
	notions,	
	ff,
	mm,
	primitives,
	events,
	complexity,
	asymptotics,
	keys]{cryptocode}

\definecolor{light-gray}{gray}{0.95}

\newtheorem*{assumption*}{\assumptionnumber}
\providecommand{\assumptionnumber}{}


\newtheorem{definition}{Definition}
\newtheorem{theorem}{Theorem}

\algdef{SE}[SUBALG]{Indent}{EndIndent}{}{\algorithmicend\ }%
\algtext*{Indent}
\algtext*{EndIndent}

\ifCLASSINFOpdf
\else
\fi

\begin{document}

\title{Extending On-chain Trust to Off-chain --\\ Trustworthy Blockchain Data Collection using Trusted Execution Environment (TEE)}
\author{Chunchi~Liu,
	Hechuan~Guo,
	Minghui~Xu,
	Shengling~Wang,~\IEEEmembership{Senior Member,~IEEE,}
	Dongxiao~Yu,~\IEEEmembership{Senior Member,~IEEE,}
	Jiguo~Yu,~\IEEEmembership{Fellow,~IEEE,}
	and~Xiuzhen~Cheng,~\IEEEmembership{Fellow,~IEEE}
	
	\IEEEcompsocitemizethanks{
		
		\IEEEcompsocthanksitem C. Liu was with the Department of Computer Science, The George Washington University and now with Ernst \& Young. E-mail: liuchunchi@gwu.edu \hfil\break
		\IEEEcompsocthanksitem  H. Guo, M. Xu (corresponding author), D. Yu, and X. Cheng are with the School of Computer Science and Technology, Shandong University. E-mail: \{mhxu, dxyu, xzcheng\}@sdu.edu.cn, ghc@mail.sdu.edu.cn \hfil\break
		\IEEEcompsocthanksitem S. Wang is with Beijing Normal University. E-mail: wangshengling@bnu.edu.cn \hfil\break
		\IEEEcompsocthanksitem J. Yu is with Qilu University of Technology. E-mail: 19001922@qq.com \hfil\break
		%
	}
}

\IEEEtitleabstractindextext{

\begin{abstract}
Blockchain creates a secure environment on top of strict cryptographic assumptions and rigorous security proofs. It permits on-chain interactions to achieve trustworthy properties such as traceability, transparency, and accountability. However, current blockchain trustworthiness is only confined to on-chain, creating a ``trust gap'' to the physical, off-chain environment. This is due to the lack of a scheme that can truthfully reflect the physical world in a real-time and consistent manner. Such an absence hinders further blockchain applications in the physical world, especially for the security-sensitive ones.

In this paper, we propose a framework to extend blockchain trust from on-chain to off-chain, and take trustworthy vaccine tracing as an example scheme. Our scheme consists of 1) a Trusted Execution Environment (TEE)-enabled trusted environment monitoring system built with the Arm Cortex-M33 microcontroller that continuously senses the inside of a vaccine box through trusted sensors and generates anti-forgery data; and 2) a consistency protocol to upload the environment status data from the TEE system to blockchain in a truthful, real-time consistent, continuous and fault-tolerant fashion. Our security analysis indicates that no adversary can tamper with the vaccine in any way without being captured. We carry out an experiment to record the internal status of a vaccine shipping box during transportation, and the results indicate that the proposed system incurs an average latency of 84 ms in local sensing and processing followed by an average latency of 130 ms to have the sensed data transmitted to and been available in the blockchain. 
\end{abstract}

\begin{IEEEkeywords}
	Blockchain; Trusted Execution Environment; Physical Traceability; Vaccine Tracing.
\end{IEEEkeywords}

}

\maketitle

\section{Introduction}

\IEEEPARstart{B}{lockchain} provides a secure environment that permits certain interactions within it to be trustworthy. It is built on cryptographic assumptions and proofs to guarantee objective security and trustworthiness. However, its trustworthiness currently is confined within the on-chain environment only, and can be hardly extended into the off-chain physical world. This is manifested by the lack of strong guarantees to ensure that blockchain data truthfully reflects the physical object it is bound, and that blockchain control instructions are always correctly deployed to the corresponding devices. This on-chain/off-chain ``trust gap''  significantly undermines the wide adoption of blockchain in security-sensitive physical world applications.  

Taking vaccine tracing as an example -- vaccine is critical to protect citizen health, thus keeping its physical distribution process highly traceable, transparent, and trustworthy is of paramount significance to public interest. It is one of the critical scenarios where digital data collected for monitoring the vaccine transportation should strictly reflect the physical world activities. However, we have seen many tragic cases where the vaccines were damaged or counterfeited due to the lack of highly trustworthy monitoring systems. Data were mostly manually entered, thus we don't know if the vaccines were \textit{actually} exposed to hazardous environments, defective or tampered with. During a surprise inspection to the rabies vaccine produced by Changchun Changsheng Ltd, China’s National Drug Administration (NDA) found more than 250,000 doses of vaccine were out of the production standards and the online data did not match with the actual product \cite{changshengNYT}\cite{changshengNature}. The Centers for Disease Control and Prevention of America (CDC) mandates the Hepatitis B vaccine to be stored between 2\textdegree{}C and 8\textdegree{}C. However, Nelson \textit{et al.} found that it was common in Indonesia that the vaccine was inadvertently frozen and thus harming its bioactivity, though the use of trusted vaccine vial monitors could significantly reduce the risk of heat-damage of the vaccine \cite{nelson2004hepatitis}. This ``physical intraceability'' challenge is common in many other security-sensitive applications in real life, and blockchain is not the silver bullet to the problem due to the current trust gap between on-chain environment and off-chain physical world. It calls for a general solution that can trustfully align ground truth on and off blockchain.

To address this issue, we resort to a trustworthy sensing scheme that 
consists of 1) a TEE-enabled trusted environment monitoring system that periodically senses the physical environment, creates anti-forgery data records, and uploads the records to blockchain; and 2) a consistency protocol to upload the records in real-time from the TEE system to the blockchain in a truthful, consistent, continuous and fault-tolerant fashion. We make careful security designs to ensure that no adversary can forge a fake data as all are verified authentic and that no adversary can physically violate the security requirements without being captured. 

 The contributions of this paper are summarized as follows.

\begin{enumerate}
\item We fully implement a trusted environment monitoring system to support trustworthy vaccine tracing in the real world-- we call it \textit{physical traceability}. We develop the system using the low-cost Cortex-M33 Trusted Execution Environment (TEE) microcontroller (MCU), which enforces strict physical isolation, thus achieving high security guarantee while keeping the system practically cheap for edge applications. We also justify why this TEE controller is the most secure \textit{and} cheap choice by comparing it with the mainstream products such as Cortex-A and Intel SGX.

\item 
We extend trust from on-chain to off-chain by designing a protocol to establish a truthful, consistent, continuous and real-time mapping between the physical object and its digital entity. We analyze the time upper bound for achieving data consistency and continuity considering general blockchain systems.

\item Our system is invulnerable to remote threats. Unless physically destroyed, any tampering with the system or the monitored physical objects is to be recorded with non-repudiation and traceability. We further provide a fault-tolerant and error recovery solution against record losses caused by transmission failures without relaxing any security requirement. We lastly discuss existing attacks against different TEE platforms, and why our system is prone to these attacks.
\end{enumerate}

The rest of this paper is arranged as follows. Section \ref{sec:def} presents the necessary definitions and basic assumptions. Section \ref{sec:PPE} outlines the vaccine tracing system under our consideration to illustrate our design. Section \ref{sec:main} provides the in-depth technical details of and comprehensive analyses on our TEE-enabled trusted environment monitoring system and the corresponding consistency protocol. Section \ref{sec:exp} demonstrates a full system implementation and reports our experimental results. Section~\ref{sec:secfuture} discusses the extensibility and possible limitations of our framework and outlines our future research to mitigate the limitations. We finally summarize the related work in Section \ref{sec:related} and conclude the paper in Section \ref{sec:conclusion}.

\section{Definitions and Models}
\label{sec:def}
In this paper, we intend to extend the on-chain trust to off-chain, and establish a digital entity that describes a physical object in a truthful, real-time and continuous manner. To achieve this goal, we need to define what is a physical object and how its status is fractured into physical events of interests, how these events are captured and stored in blockchain in the form of digital timed records, and how a consecutive sequence of timed records eventually forms a digital entity. We lastly define the trustworthy standard of a digital entity, which includes truthfulness, real-time consistency and continuity.


\subsection{Basic Definitions}
\label{sec:defs}

A \textit{physical object} is a minimal unit in the physical world whose status interests us. It could be one or a batch of homogeneous things, or even a small physical area.

A \textit{physical event of interest} (hereafter abbreviated as an event) is a status of a physical object at a certain moment of time in which the public is interested in. An event can be  1) a temporal checkpoint in a periodical monitoring setting, or 2) a sudden change in sensed data (``change-of-status'') that exceeds a maximum tolerable range in an event-triggered monitoring setting.

A \textit{digital timed record} (hereafter abbreviated as record) is a captured result of an event by a set of sensors. A record is denoted as $(t, C, \pi)$, where $t$ is the timestamp when this record is created, $C$ is a combination of different types of sensor outputs such as binary, numerical, categorical or multimedia, and $\pi$ is a proof of authenticity of the record or (partly) the trustworthiness of the data. By this way one can see that  a record is a digital description towards an event. Although there exists two types of events -- periodical checkpoints and sudden change-of-status, in this paper we refer the former as \textit{ records} while the latter is processed as an independent alarm message. This is because we make use of periodically sensed records as the main source of evidence to describe a physical object. 


A \textit{digital entity} $\alpha$ can then be formulated as a chain of chronologically consecutive records. It fully captures all events of a physical object and describes its state transitions. 
Under the periodical monitoring setting, the continuity is guaranteed by sensing and reporting records periodically once every $\Delta t$ time. Under the event-triggered monitoring setting, an entity is constructed with significant change-of-status records. 
One can see that both the periodical and event-driven monitoring settings are desirable in practice, thus in this paper
we consider a hybrid one -- we use a periodical monitoring setting to record data but permit alarm messages to be sent out at anytime upon receiving a significant ``change-of-status'' event. A digital entity $\alpha$ can be denoted as
\begin{equation}
\label{eq:digitalentity}
\alpha = (t_0, C_0, \pi_0), \cdots, (t_i, C_i, \pi_i), \cdots
\end{equation}

With this hybrid setting, one needs to enforce the following three mandatory requirements on a digital entity in order to confidently believe that the entity does depict a physical object with trustworthiness: i) all associated records should truthfully reflect reality of the physical world; ii) each record should be captured in real-time after an event occurs; and iii) the records should be uploaded to blockchain without loss or significant time delay. They can be formalized into the following definition:
\begin{definition}[Trustworthiness of a digital entity]
\label{def:trust}
	We say a digital entity is \textit{trustworthy} if it satisfies all the following three requirements:
	\begin{itemize}
		\item \textbf{Truthfulness}: each record is generated by a secure system or program that is tamper-proof.
		\item \textbf{Real-time consistency}: For any event occurred in the physical world at time $t$, it takes at most $\delta$ more time for the blockchain system to return the same timed event, where $\delta$ is a small real number that is application-specific.
		\item \textbf{Continuity}: The time interval between two consecutive records of an entity is confined within $\hat{\delta}$, where $\hat{\delta}$ is a real number that does not significantly deviate from $\Delta t$, the sensing interval.
	\end{itemize}
\end{definition}


Regarding a physical event of interest, we judge whether or not the event is ``normal'' by comparing the sensing result $C$ against a pattern $\Gamma$, which defines the correct pattern of normal sensing data. In practice, $\Gamma$ and $C$ share the same structure, and the checking can be done by computing the deviation between $C$ and $\Gamma$. The combination of $C$ and $\Gamma$ can describe the basic detail of a physical event and how desirable or how legal it is. More discussions can be found in Section \ref{sec:PPE}. As one can see, these definitions mandate a truthful, robust, and highly-responsive system that continuously enforces monitoring towards a physical object, and then uploads the data to blockchain with low latency. This process must be efficient enough to satisfy the real-time requirement. We realize this goal by designing a consistency protocol and implementing a trusted environment monitoring system leveraging TEE, as detailed in Section \ref{sec:main}.

We lastly define a general blockchain system as a distributed network consisting of $n$ processors (or nodes), in which at most $f$ of them may be \textit{byzantine} (capable of exhibiting arbitrary behaviors), and the other $n-f$ are \emph{correct} (faithfully following the protocol without any fail-stop, omission or byzantine failure).  We consider modeling a general blockchain using the partially-synchronous model \cite{dwork1988consensus}\cite{1984Consensus}\footnotemark{}, \footnotetext{Solving fault-tolerant, deterministic consensus in asynchronous networks has been proven as an impossibility \cite{fischer1985impossibility}. While many works claimed achievement of asynchronous consensus, they in fact adopted a weak, or partially-synchronous assumption, requiring the network to synchronize periodically (e.g., Delegated PoS \cite{eoswhitepaper}) or under an exist-but-unknown-a-priori upper bound (e.g., PoW \cite{nakamoto2019bitcoin}).} \footnotetext[2]{In Dwork \textit{et al.}'s original definition \cite{dwork1988consensus}\cite{1984Consensus}, the \textit{defined} GST is unknown to the processors but a message communication upper bound $\Delta$ exists after a GST is announced and this upper bound stays valid during a time interval $\mathbf{L}$, meaning the network achieves temporary synchrony. This means that the corresponding model excludes the case where $\Delta GST$ is a constant and GST is known. In our definition, we simplify and generalize this model by setting no constraint on GST, but $\Delta GST$ can be a constant or a random variable with a known, finite upper bound, and could potentially follow a distribution (but we do not mandate it). Our GST is just a moment at which a new block production process is initiated.} which permits periodical synchronization at Global Stabilization Times (GSTs). After each GST, the blockchain system is temporarily synchronized for $\mathbf{L}$ units of time. During each [GST$_i$, GST$_i$+$\mathbf{L}$) interval, a new round of consensus is initiated, executed and completed by committing a new block to all correct nodes -- we call this a \textit{block production process}. In real life, a GST is usually the moment at which a request or a command to initiate a new block production process is issued. The blockchain system can receive new transactions at anytime, but these transactions need to wait to be committed; thus we assume that all data received during [GST$_i$+$\mathbf{L}$, GST$_{i+1}$) are to be included in the next block at depth $i+1$. The interval between two GSTs is denoted by $\Delta GST$. Without loss of generality, $\Delta GST$ can be a constant value or a random variable with a known finite upper bound -- respectively depicting cooperative consensus that produces blocks regularly and deterministically (such as Delegated Proof-of-Stake), and competitive consensus that produces blocks with a degree of uncertainty (as in Proof-of-Work and other mining-based consensus algorithms)\footnotemark{}.




\subsection{Trust Model}

In this study, we trust that the TEE hardware is secure against any long-range vulnerability, and that blockchain cannot be manipulated by an adversary.

\textbf{Blockchain:} Generally, a blockchain system is \emph{secure} if its consensus output is secure against adversaries' malicious manipulations and the consensus output does not change in any node's view after finalization. These two requirements can be described as: 1) the protocol is proved secure and is correctly implemented, and is only vulnerable to node byzantine faults; and 2) the number of existing byzantine nodes does not exceed the maximum number $f$ tolerable by the blockchain system. A blockchain system is \emph{available} if all non-faulty data can be included in the blockchain within a finite time interval. Our blockchain is assumed to be both secure and available. With the fact that most non-trivial blockchain systems or consensus algorithms provide such guarantees, one can reasonably make this trust assumption. 

\textbf{TEE:} A Trusted Execution Environment (TEE) physically separates a secure zone from a non-secure zone (a.k.a. the rich-environment zone). Programs in the secure zone can only be called by the non-secure zone but cannot be modified or explicitly inspected. It is a general consensus among the security community that programs within the secure zone are invulnerable against long-range tampering. There is one master secret key $mas\_{sk}$ that is unique to each TEE system and can be used to exclusively authenticate this trusted device to the public. This master secret key cannot be explicitly retrieved or tampered with. We assume the TEE hardware has a public/private key pair, denoted by $TEE\_{pk}$ and $TEE\_{sk}$, as a blockchain client needs to be implemented in TEE. Note that we consider physical damages towards TEE out of scope.

The controversy on TEE security lies in different TEE architectures, and we choose the Arm Cortex-M series TEE microcontroller since other choices of TEE such as the Cortex-A series and Intel SGX, which are prevalent in most smartphones and servers, do not strictly enforce physical isolation on the secure zone, making the system insecure  as evidenced by the many reported security attacks against program integrity \cite{murdock2020plundervolt} and private data secrecy \cite{reinbrecht2016side}. Cortex-M microcontrollers enforce strict physical isolation through mandatory manual configuration and only side channel attacks in a macro scale (such as power analysis) can be done with high time cost (estimated more than 100 to 300 seconds), while the attack results can only compromise some secrecy but not program integrity. In later discussions we explain how one can thwart this attack through memory masking or simply prohibiting possible attack windows from lasting longer than 100 seconds. We provide detailed analysis and comparison studies on state-of-the-art attacks and defenses on TEE security in Section \ref{sec:secTEE}.


\subsection{Notations}

We organize all notations and symbols in Table 1.

\begin{table*}
	\caption{Summary of Notations}
	\label{table:summary:notation}
	\centering
	\begin{tabular}{c|c}
		\hline
		\textbf{Symbol} & \textbf{Description} \\
		\hline
		$t, \Delta t$ & the time stamp at which a record is created; the sensing interval \\ \hline
		$T, \Delta T$ & the moment a record is available in blockchain; time difference between two consecutive $T$s\\ \hline
		$C$ & a combination of sensor outputs; in this paper we have: \\
		$(L, K, x, y)$ & brightness level, temperature, and GPS location $(x,y)$\\ \hline
		$\Gamma$ & a valid range of sensor output, namely pattern; in this paper we have:\\
		$\Gamma_P = \{0/1\}$ & binary photosensor output, 0 or 1 for dark or bright\\
		$\Gamma_K=  [K_{min}, K_{max}]$ & temperature sensor output, minimum and maximun permitted temperature\\
		$\Gamma_G=\{(X_i,Y_i), r_i\}$ & GPS locator output, a list of checkpoint locations and the permitted deviation radius\\ \hline
		$\pi$ & proof of data authenticity and trustworthiness\\ \hline
		$n, f$ & number of all nodes and that of byzantine nodes in blockchain\\ \hline
		GST, $\Delta GST$ & Global Stabilization Times; time difference between two consecutive GSTs \\ \hline
		$\mathbf{L}$ & guaranteed synchronized period after GST \\ \hline
		$F_d$ & backup queue for unsuccessfully uploaded data \\ \hline
		$mas\_{sk}, TEE\_{pk}, TEE\_{sk}$ & globally unique master secret key of each TEE and its derived public and private key pair \\
		$sym\_{sk}$ & derived session symmetric key between TEE and a proxy blockchain client\\
		$recp\_{pk}, recp\_addr$ & proxy blockchain client public key and network address\\ \hline
		$\phi$ & idle time from a transaction joining the blockchain pool to the moment next block production process starts\\
		$l_1$ & time interval of a  block production process\\
		$l_2$ & time interval between block production process finishes to synchronization period expires\\ \hline
		$\epsilon_1$ & time delay of local data sensing in the TEE system\\
		$\epsilon_2$ & time delay of transmission between the TEE system and the remote blockchain agent\\
		$\epsilon_3$ & time delay of block commit from  the remote blockchain agent  to all the blockchain nodes \\
		$\epsilon_s$ & clock difference between the TEE system and the remote blockchain agent\\
		$\epsilon$ & total time delay from a physical event's occurrence to the digital record ready in blockchain \\ \hline
		$\delta$ & upper bound of time delay between any event's occurrence to its corresponding record's global commit \\
		$\hat{\delta}$ & upper bound of time delay between any two consecutive records \\
		\hline
	\end{tabular}
\end{table*}

\section{A Vaccine Tracing System Example}
\label{sec:PPE}

In this section we introduce the vaccine tracing system as an example to demonstrate our framework of extending trust from on-chain to off-chain physical world. More specifically, our scheme is designed to make sure that the environment within the vaccine shipping box is monitored trustworthily in realtime. 

We consider that an important vaccine is stored inside an insulation box. This box must be transported from A to B via an approved route. The box is sealed and is prohibited from being opened. The temperature inside the box must be kept low and stay stable in order to preserve the biological activity of the vaccine. We start by defining possible situations that may violate the security requirements:
\begin{enumerate}
    \item the box is opened (may destroy or replace the vaccine),
    \item the temperature within the box is abnormal (may nullify the vaccine biological activity),
    \item the transportation route is deviated from the predefined one (same as 1),
    \item records may be lost (same as 1).
\end{enumerate}

To capture these violations, we place a photosensor, a temperature sensor, and a GPS locator inside the box. As discussed earlier, pattern $\Gamma$ needs to be predefined to determine whether the data $C$ complies with the security requirements or violates them. In this vaccine tracing system, the photosensor should always output 0 for constant darkness, indicating the box is sealed, thus we can define $\Gamma_P = \{0\}$; the actual binary data of the photosensor is denoted by $L$, and  $L=1$ if the ambient brightness is greater than the luminous threshold $\theta$ and $L=0$ otherwise. The temperature sensor should always report a steady internal temperature ranging from $K_{min}$ to $K_{max}$, so $\Gamma_T = [K_{min}, K_{max}]$; the actual capture of the temperature sensor is $K$ degrees. The GPS sensor records the physical shipping route of the box, and should follow a geographic pattern from the origin, along a reasonable path, to the final destination. We define $\Gamma_G = \{(X_i,Y_i), r_i\}$, which includes a series of checkpoints $(X_i,Y_i)$ by latitudinal and longitudinal coordinates and a safe radius between the box and the closest checkpoint. The latitudinal and longitudinal coordinates captured by the GPS sensor are respectively denoted as $x, y$. If the Euclidean distance between each location upload $(x_i,y_i)$ and the nearest checkpoint $(X_i,Y_i)$ is less than $r_i$, the box is considered in a safe route. The actual sensing data to be uploaded for each sensing event is denoted by:
\begin{equation}
    C = (L, K, x, y)
\end{equation}

We would like to point out that one could use more types of sensors to monitor the box during transportation. For example, one can use a smart lock, a motion sensor, and a humidity sensor, to enhance the monitoring effect of the vaccine box. These sensors are not hard to add-on, as one can see from next section that our design philosophy can be easily applied to them. This study represents our exploratory effort towards trustworthy vaccine shipment and for demonstration purpose we focus on a simple example and present the most essential components, leaving other opportunities to future real world system developers.


\section{Main Scheme: Trust Extension from On-Chain to Off-Chain Physical World}
\label{sec:main}

In this section, we detail the design and implementation of our trust extension scheme using the vaccine tracing as a case study example. Our system consists of a TEE-enabled trusted environment monitoring system and a consistency protocol for uploading data from the TEE system to blockchain in an orderly fashion. As discussed earlier, vaccines may suffer from counterfeit, physical damage, being unsealed or replaced by fake ones during transportation. Nevertheless, it is extremely challenging to seamlessly monitor their transportation in a trustworthy way in practice. This ``physical intraceability'' is common in many security-sensitive real world applications. We propose a system that permits secure and trustworthy vaccine shipping, which can capture any tampering and violation to the physical object, guaranteeing non-repudiation and traceability. We also present a fault tolerance mechanism that can recover from lost packets due to transmission failures, ensuring seamless monitoring of the vaccine transportation.

\begin{algorithm}[!htb]
\caption{Main System Utilities}
\label{mainsystem}
\begin{algorithmic}[1]
\State //\textbf{Symmetric Key Derivation and Distribution}
\State \textbf{Function} \texttt{KDF}($mas\_{sk}, recp\_{pk}, recp\_addr$) \\
//Master secret key of the system, recipient (blockchain client) public key, recipient address. \\
//Symmetric encryption for efficiency, asymmetric signature for public blockchain verification.
\Indent
\State $sym\_{sk}$ = $Enc_{\text{AES}}$($mas\_{sk}||recp\_{pk}||\text{TRNG}()$)
\State //Derive session key $sym\_{sk}$ from master key $mas\_{sk}$
\State $distribution = Enc_{\text{RSA}}(sym\_sk, {recp\_{pk}})$
\State $signature = Sign_{\text{RSA}}$($distribution$, ${TEE}\_{sk}$)
\State 4Gsend($distribution$, $signature$, $recp\_{addr}$)
\EndIndent
\State \textbf{return} $sym\_{sk}$
\\

\State //\textbf{Retrieving Sensor Data}
\State \textbf{Function} \texttt{Sensors\_get}($\theta$)
\State \textbf{Initialization:} $L=0, K=0, x=0, y=0$
\Indent
\If {($app\_ambient\_lum() > \theta$)} $L$=1  \\
\qquad //Hardware level logic
\EndIf
\State $K$ = $app\_temp\_get()$
\State $(x,y)$ = $app\_gps\_get()$
\EndIndent
\State \textbf{return} $C=(L,K,x,y)$
\\

\State //\textbf{Violation Detection}
\State \textbf{Function} \texttt{Violation\_check}($C, \Gamma_P,\Gamma_T,\Gamma_G$)
\State \textbf{Initialization:} ${flag}_P={flag}_T={flag}_G=0$, $msg = null$
\Indent
\If {($C.L!=\Gamma_P$)} \\
\qquad {${flag}_P$=1, $msg.append (\text{``Box opened''})$ }
\EndIf
\If {($C.K> \Gamma_T.K_{max}\ ||\ C.K< \Gamma_T.K_{min}$)}\\
\qquad {${flag}_T=1$, $msg.append (\text{``Abnormal Temperature''})$ }
\EndIf
\If {($\min \{ dist((C.x, C.y),(\Gamma_G.X_i, \Gamma_G.Y_i)) - \Gamma_G.r_i\} > 0$)} \\
\qquad {${flag}_G=1$, $msg.append(\text{``Route Deviated''})$ }
\EndIf
\If {(${flag}_P||{flag}_T||{flag}_G$)}
\qquad  \State $t$ = $get\_sys\_clock()$, $msg.append(t)$
\qquad \State 4Gsend($msg, Sign_{\text{RSA}}(msg, TEE_{sk}), recp\_addr$)
\EndIf
\EndIndent
\State \textbf{return} 0
\\

\State //\textbf{Send Data Packets via 4G} 
\State \textbf{Function} \texttt{4GSend}($msg, recp\_addr$)
\Indent
\State $app\_4G\_send(msg, recp\_addr)$
\If{(not receiving $\texttt{ACK}$)} {\textbf{return} `fail'}
\State \textbf{else return} `ok'
\EndIf
\EndIndent

\end{algorithmic}
\end{algorithm}

\subsection{TEE-Enabled Trusted Environment Monitoring}
We first introduce our full system of TEE-enabled trusted environment monitoring system. Fig.~\ref{mainsystem} presents the major utilities. We develop our system from the bare metal level for best security guarantee and efficiency/cost performance.

As discussed earlier, the secure zone inside a TEE hardware has the highest security privilege through physical isolation. Programs implemented in the secure zone can only be called by non-secure zone through a callable-API, and cannot be modified or inspected by the non-secure zone. Therefore, the secure zone should execute security-critical tasks such as trusted data collection, pattern extraction, encryption and decryption. The non-secure zone can perform non-security tasks such as user interface or packet routing. To initialize, we manually configure the Secure Attribution Unit (SAU) in the secure zone to partition the secure-zone memory from non-secure zone memory.  To make the system function as expected, we first develop drivers of security-critical sensors within the secure zone and directly connect them to the corresponding devices by wire; then we assign higher system interrupt priorities to security-critical tasks in order to achieve better real-time performance. Fig.~\ref{fig:blockdiagram} demonstrates the system block diagram with basic components introduced as follows.

\begin{figure}[!htb]
	\centering
	\includegraphics[width=0.45\textwidth]{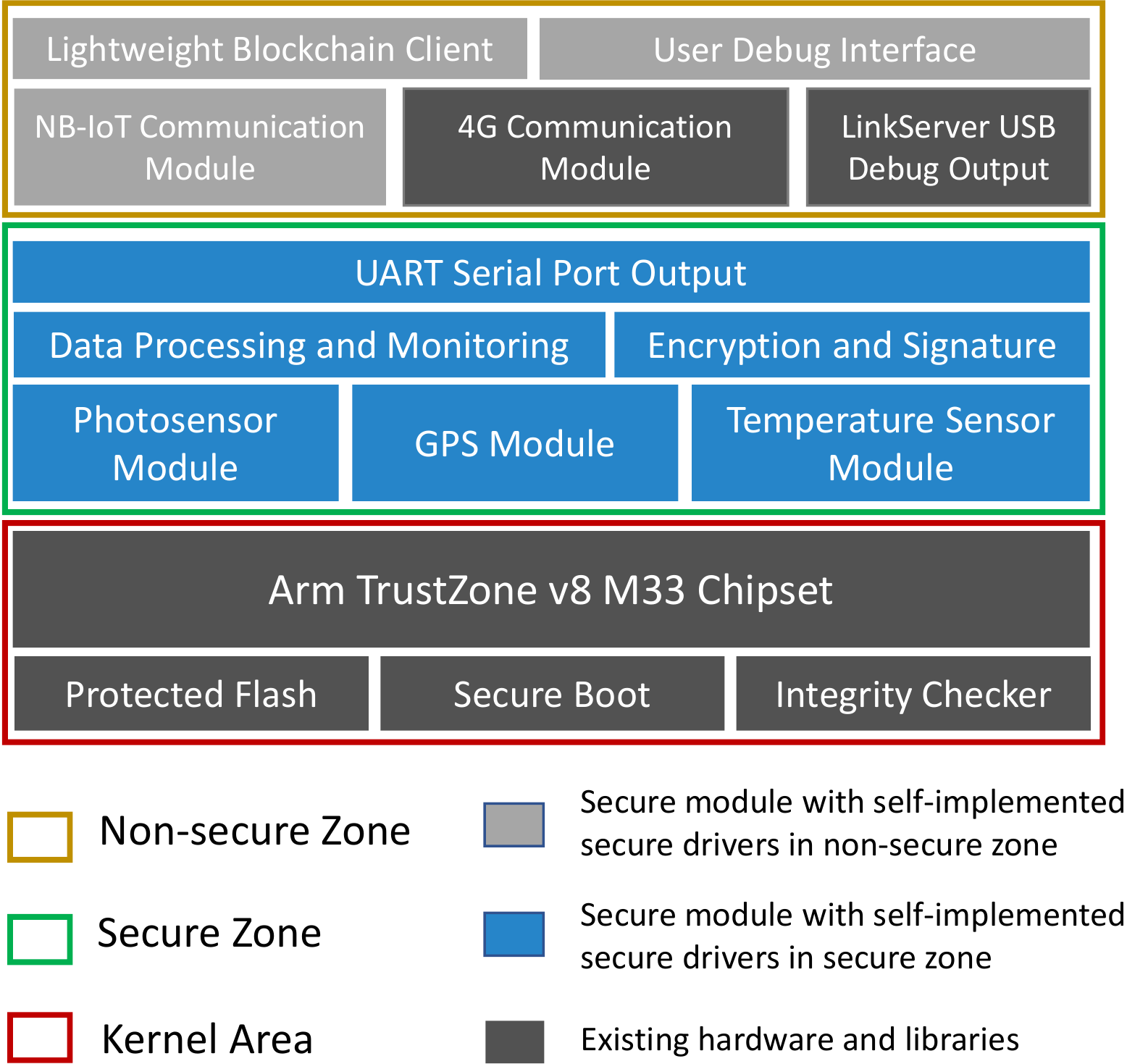}
	\caption{Block Diagram of the System}
	\label{fig:blockdiagram}
\end{figure}

At the lowest kernel level, a secure boot module and an integrity checker are first activated after booting. These two modules validate the current image (burnt-in executable) by computing its RSA signature and comparing it with the pre-stored correct one. If no corruption is found, the CPU control is transferred to the secure zone user space. The Protected Flash Region (PFR) stores the master secret key $mas\_{sk}$ and the key pair $TEE\_{pk}$ and $TEE\_{sk}$, and must be checked for integrity by secure boot. If passing the check, which means that $mas\_{sk}$ and the key pair in PFR are intact, they are retrieved from PFR and loaded into the secure zone RAM. One can see that these keys cannot be faked, corrupted, or stolen from the PFR and the secure zone. The master secret key $mas\_{sk}$ is used to derive session keys for AES encryption while the public key pair  is used for message authentication and signature verification. The key pair $TEE\_{pk}$ can be registered through its owner onto the blockchain or managed through public key infrastructure.

At the secure zone we develop and assemble the drivers of the photosensor, GPS, camera, and temperature sensor in C language on our own. We also define critical logic and parameters such as the raw data collecting procedure, buffer size, inner capture frequency, debug procedures, and so on. This is reflected by \textbf{Function} \texttt{Sensors\_get($\theta$)} in Algorithm \ref{mainsystem}, where $\theta$ is a threshold to ensure that  the photosensor returns binary 1 if and only if its reading is above $\theta$.

The data processing and monitoring module collects data from the sensors as $C$ then compares the data with the predefined legit pattern $\Gamma$. If this checking finds any violation against $\Gamma$ defined in Section \ref{sec:PPE}, an immediate alarm message along with the current timestamp is generated. Note that this alarm message is directly sent out and does not affect normal data uploading. This is shown as \textbf{Function} \texttt{Violation\_check}($C, \Gamma_P,\Gamma_T,\Gamma_G$) in Algorithm \ref{mainsystem}. Next we retrieve the current system clock $t$, and packs all into a correct data structure $(t, C, \pi)$, where $\pi$ is the publicly verifiable signature computed using $TEE\_{sk}$ to authenticate the trusted source of data $C$. Finally we check if there exist any lost history data, and if so, include them as $F_d$; then we sign and encrypt $F_d||(t, C, \pi)$ with the session key $sym\_{sk}$ and call \texttt{4GSend()} presented in Algorithm \ref{mainsystem} to upload the data to the blockchain. This procedure is summarized by Algorithm~\ref{mainprotocol} and the flow chart of the data uploading procedure is illustrated in Fig.~\ref{fig:flowchart}.


\begin{figure*}[!htb]
	\centering
	\includegraphics[width=0.96\textwidth]{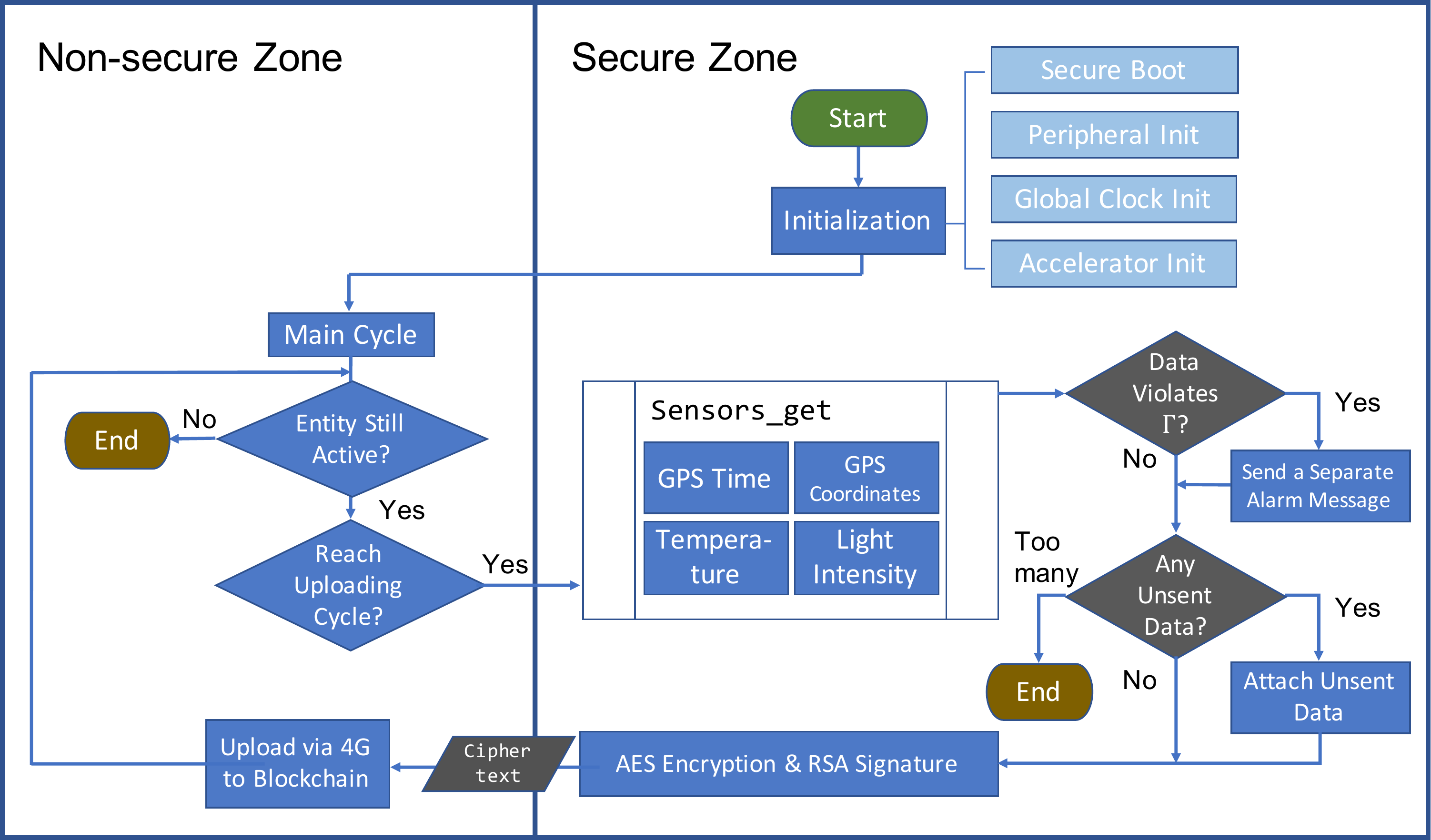}
	\caption{Flowchart of the System}
	\label{fig:flowchart}
\end{figure*}

The encryption and signature module performs encryption, signature signing and verification. Recall that the system has a master secret key $mas\_{sk}$ and a public/private key pair $TEE\_{pk}$ and $TEE\_{sk}$. To save computational resource, we need a symmetrical session key for encryption with AES. This symmetric key $sym\_{sk}$ can be derived as an AES output of encrypting the concatenation of the master secret key $mas\_{sk}$, the recipient public key $recp\_{pk}$, and a  random number which is generated by the embedded True Random Number Generator (TRNG). AES encryption is a common source of randomness to derive symmetric keys. Since it is CCA-secure\footnote{Using Authenticated Encryption modes such as GCM.}, attackers cannot have more than negligible probability to infer $mas\_{sk}$ providing $sym\_{sk}$. We then encrypt the symmetric key $sym\_{sk}$ once using the remote blockchain agent's public key to securely deliver $sym\_{sk}$. Note that it is the remote blockchain agent's liability to keep the symmetric key secret. This function is shown in \textbf{Function} \texttt{KDF}($mas\_{sk}, recp\_{pk}, recp\_addr$).

Finally in the secure zone, a Universal Asynchronous Receiver/Transmitter (UART) module receives  messages and sends them to the non-secure zone. 
The 4G or NB-IoT communication module, or the LinkServer local USB debug output, receives the data from the UART module. For time-sensitive applications the 4G module can be used while for energy-efficiency-sensitive applications the NB-IoT module can be used. The lightweight blockchain client in TEE provides the current remote blockchain agent server's address to which data can be redirected after entering the Internet. A user debug interface connects UART to USB and permits output via PC at the specific IDE. Fig.~\ref{fig:TEEArchitecture} shows another view of the system.

\begin{figure}[!htb]
	\centering
	\includegraphics[width=0.45\textwidth]{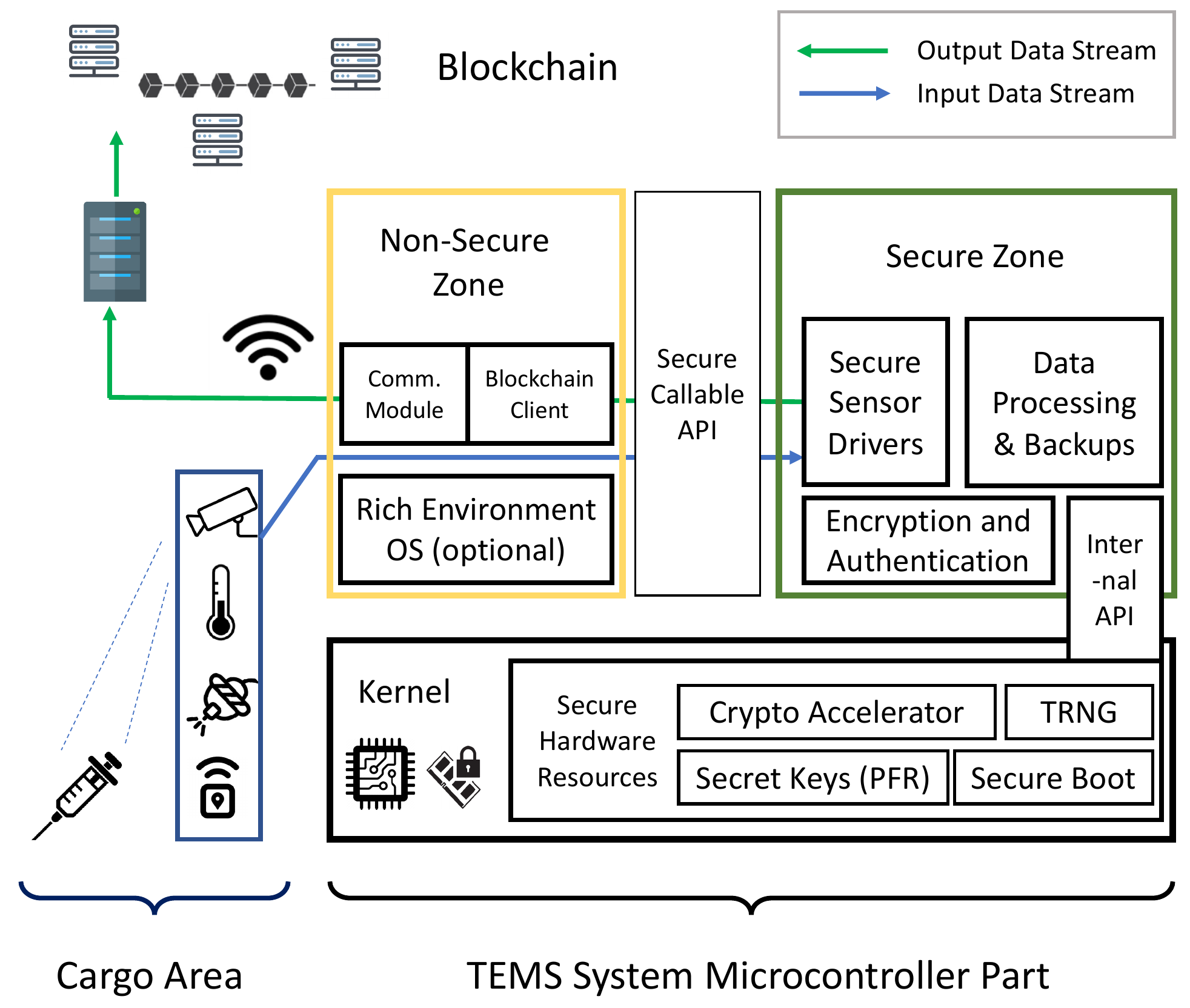}
	\caption{Abstract Architecture of System}
	\label{fig:TEEArchitecture}
\end{figure}

We notice that there exist other works that use TEE to perform trustworthy operations. However, the adopted TEE hardware chipsets either are prohibitively expensive as the more-easy-to-implement but expensive ones such as the Intel SGX are adopted, or make use of the ones such as the Cortex-A series that do not strictly enforce the secure zone physical isolation. In this study, we implement our system using the Cortex-M series TrustZone Chipset as it is the first series that physically divides secure zone from the non-secure zone. Not surprisingly,  we face a great engineering challenge and overhead because there exist very few usable libraries, kernels, and operating systems developed for the Cortex-M series. As a result, we build the whole system from the bare metal level using C and Assembly. To the best of our knowledge, we are the first to implement such a trustworthy system using the challenging yet cheap Cortex M33 MCU, which is priced around \$40, while most other Intel SGX products cost around \$300 {\raise.17ex\hbox{$\scriptstyle\sim$}} \$400.

\begin{algorithm}
\caption{Data Uploading Protocol}
\label{mainprotocol}
\begin{algorithmic}[1]
\State \textbf{Input:} $recp\_addr, \max\_F_d$, $\theta$, ($\Gamma_P$, $\Gamma_T$, $\Gamma_G$) //Remote blockchain agent address, max number of packet resend tolerance, threshold of photosensor, predefined security pattern.
\State \textbf{Initialization:} Initialize system hardware; activate blockchain; synchronize system clock and blockchain client with external GPS time; initialize message queue.
\State $sym\_{sk}$=\texttt{KDF($mas\_{sk}, recp\_{pk}, recp\_addr$)}
\State $F_d=$ null
\While{(It is time for periodic reporting)}
\State $C$ = \texttt{Sensors\_get}($\theta$)
\State \texttt{Violation\_check}($C, \Gamma_P,\Gamma_T,\Gamma_G$)
\State $t$ = $get\_sys\_clock()$
\State $\pi$ = $sign_{\text{RSA}}(C,{TEE}_{sk})$ //New data
\If {(!queue.empty() \&\& len(queue)$\leq \max\_F_d$)}
\While{(!queue.empty())}
\State {$F_d$.append(queue.front())
\State queue.pop()}
\State //Get the failed-to-send historical data array
\EndWhile
\Else
\If {(!queue.empty() \&\& len(queue)$ > \max\_F_d$)}
\State \textbf{return} `Exceed maximum recovery tolerance'
\EndIf
\EndIf
\State $F_d.append((t,C,\pi))$
\State $msg = Enc_{\text{AES}}(F_d, sym\_{sk})$
\State status=4Gsend($msg$, $recp\_addr$)
\If {{(status==`fail')}}
\For{(a : $F_d$)}
\State queue.push(a)
\EndFor
\EndIf
\EndWhile
\State \textbf{return} 0
\end{algorithmic}
\end{algorithm}

\subsection{Consistency and Continuity Analysis}
\label{sec:analysis}
As discussed earlier, one can see that to achieve trusted physical traceability, we must ensure that
\begin{enumerate}
    \item \textit{\textbf{Real-time consistency}}: any event that happens must be uploaded and available within $\delta$ time,
    \item \textit{\textbf{Continuity}}: a periodical monitoring at every $\Delta t$ time, and the time delay between any two consecutive records must be within $\hat{\delta}$ time.
\end{enumerate}

%
Algorithm \ref{mainprotocol} shows the pseudocode of the data uploading protocol. 
We now analyze the worst latency performance for both $\delta$ and $\hat{\delta}$ of this protocol\footnote{In our analysis, we ignore the local data processing time within a node or an agent as it is generally small and negligible.}.


Note that time latency $\delta$ is a sum of three random system delays, i.e., $\delta =\epsilon_1+\epsilon_2+\epsilon_3$, which are described as follows.
\begin{enumerate}
    \item ${\epsilon}_1$ is the delay of local data sensing in the system. It begins at the moment the system starts sensing, and ends at the time the secure zone outputs a signed ciphertext data $D=F_d||(t,C,\pi)$ as a digital record.
    \item ${\epsilon}_2$ is the delay of transmission between the TEE system and the remote blockchain agent. 
    It starts by the non-secure world of TEE receiving $D$ and ends at the time when the remote blockchain agent obtaining $D$.
    \item ${\epsilon}_3$ is the delay of synchronizing the records in the whole blockchain system. It starts by the first blockchain node (with the remote blockchain agent mentioned above) receiving $D$ and ends when all good blockchain nodes retrieving $D$ from the blockchain.
\end{enumerate}

Note that we use $\epsilon_s$ to denote the clock difference between the TEE system and the remote blockchain agent, as the latter actually proposes the record to blockchain.

Fig. \ref{fig:timeslotnormal} shows a normal operation of the system. As defined in Section \ref{sec:defs}, we assume a general blockchain under a partially synchronous model. Time $t_i$ is the local system timestamp at TEE when the event record \textit{should} be created, and $T_i$ is the time when blockchain has this event publicly available. The time gap $\Delta t$ is the desired period between two planned sensing events and $\Delta T$ is the time difference between two timed events recorded in blockchain. We know that the system latency $\epsilon_1$ can vary but is stable for a properly developed system, and the transmission latency $\epsilon_2$ should also be stable for a high-performance transmission protocol like 4G. The only variable that may significantly affect the latency is $\epsilon_3$, where $\epsilon_3=\phi+l_1$, with $\phi$ being the idle time that a record waits in the blockchain pool to the initiation of next block production process\footnote{Recall that for simplicity, we assume all received data will be included in the next block, which means that $\phi \in (0,\Delta GST]$. This can be relaxed if additional analysis is performed under different settings and assumptions.}, and $l_1$ being the time interval the block production process lasts, committing new blocks to all nodes. We call GST$+l_1$ the Block Commit Time (BCT). From BCT to the time when the synchronization period $\mathbf{L}$  expires is denoted by $l_2$, which satisfies $\mathbf{L}=l_1+l_2$. Without loss of generality, we assume that $l_1$ is a random variable in $(0,\mathbf{L}]$, which \textit{could potentially} follow a distribution (but we do not mandate it).

\begin{theorem} \label{theorem:latency1}
Assume that there is no record loss. For a sensing event started at time $t$, it takes at most $\delta = \epsilon_1+\epsilon_2+\Delta GST+\mathbf{L}+\epsilon_s$ time to have this sensed event available in blockchain.
\end{theorem}
\begin{proof}
The sensing event starts at time $t$, and after $\epsilon_1+\epsilon_2$ time the corresponding record reaches the blockchain. In the worst case, the record just misses the next block production deadline (GST) thus it must wait for $\Delta GST$ time to be included in the next block production process. Considering the worst case where the blockchain takes the longest permitted time to execute the consensus process, which is an extra $\mathbf{L}$ units of time, and the time synchronizing difference, one can get the result.
\end{proof}

\begin{theorem}\label{theorem:latency2}
Assume that there is no record loss. For two consecutive timed events recorded in blockchain, the maximum time difference is $\hat{\delta} = \Delta t + \Delta GST + \mathbf{L} +\epsilon_s$.
\end{theorem}

\begin{proof}
The maximum time difference between two records exists between a best case record processing followed by a worst case one. According to the definition, $\Delta T = t_{i+1}+\delta_{i+1}-(t_i+\delta_i) = \Delta t + \Delta \delta$, where $\Delta t$ is the planned period between two sensing actions, and in this case $\Delta \delta$ is the max record processing time difference. In the best case, $t_i+\epsilon_1^i+\epsilon_2^i$ meets the next block production deadline (GST) and the blockchain takes an ideal instant time to finish the block production process, which means $\phi_i=l_1^i=0$ and therefore $\epsilon_3^i=0$. The worst case is analyzed in \ref{theorem:latency1} where $\phi_{i+1}=\Delta GST$ and $l_1^{i+1}=\mathbf{L}$. Considering the time synchronizing difference one can get the result.
\end{proof}

\begin{figure}[h]
	\centering
	\includegraphics[width=0.42\textwidth]{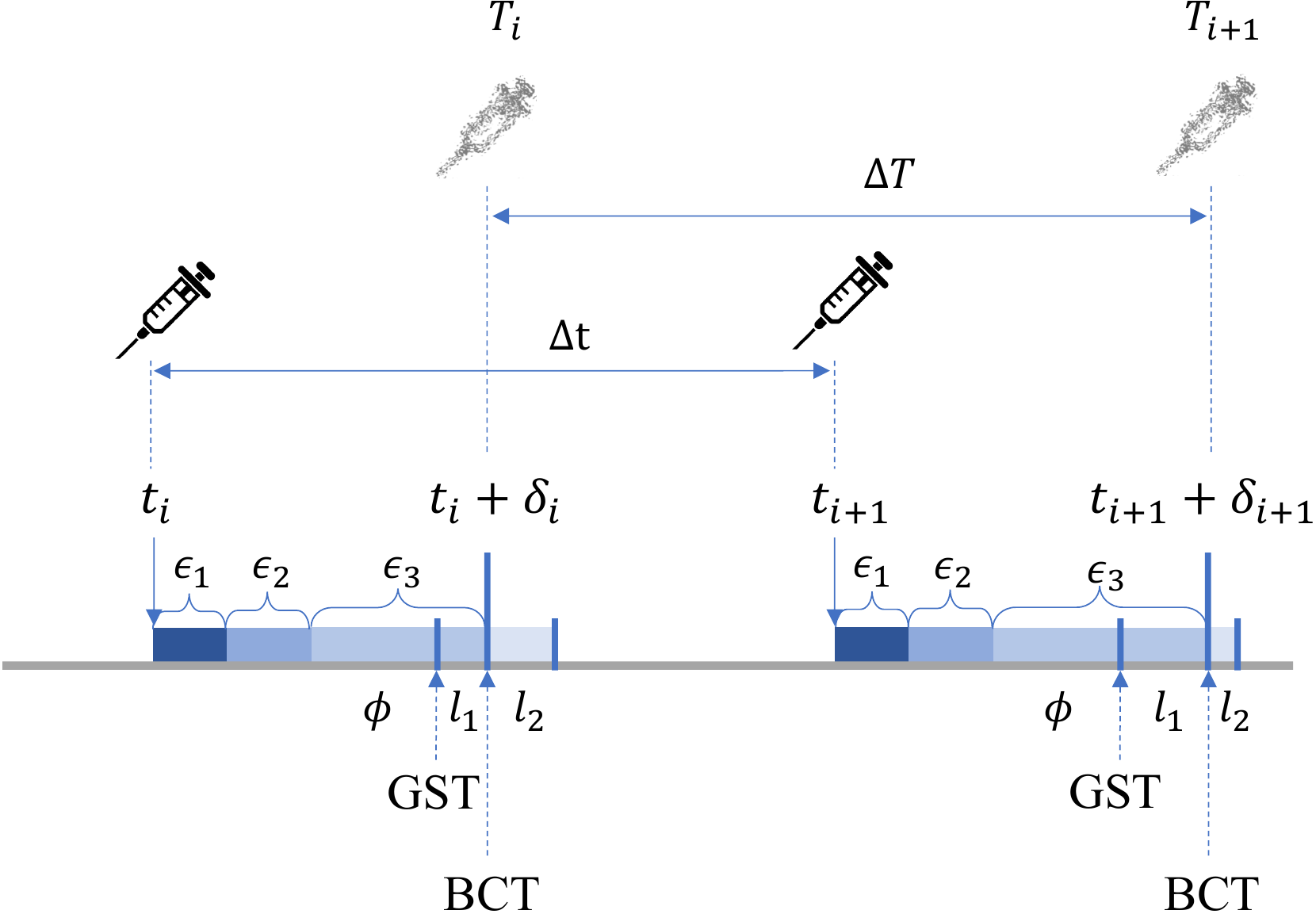}
	\caption{Data uploading illustration. Superscript $i$ and $i+1$ of $\epsilon, \phi, l_1$ and $l_2$ are omitted for conciseness; GST stands for the Global Stabilization Time and BCT stands for the actual Block Commit Time.}
	\label{fig:timeslotnormal}
\end{figure}

\subsection{Lost Record Tolerance and Recovery}
\label{sec:recovery}
The whole system can malfunction in many ways; thus we need to consider fault tolerance as well as recovery. The malfunction or halt of the TEE system itself must require manual inspection and restart, and that of the blockchain is out of the scope of this paper. Therefore we consider fault tolerance and recovery for the transmissions between the TEE system and the blockchain system, which are common in real world and doable at the software level. 
On the other hand, if communication fails and the system can prove to the blockchain that it did create the data record back in time as expected, those lost or late packets can be accepted without distrusting the device.

Let's examine Algorithm \ref{mainprotocol}  again. Recall that for each digital record $D_i=F_d||(t_i,C_i,\pi_i)$, the lost packets in $F_d$ are stored in a queue and retrieved as a list (see Lines 10-20 in Algorithm \ref{mainprotocol}).
%
If the system does not receive an \texttt{ACK} for a record from the blockchain (see Lines 23-28 in Algorithm \ref{mainprotocol}), it pushes the data into the backup queue and resends it at the next uploading cycle. Note that we allow a maximum of $\max\_F_d$ number of consecutive missed packets to be backed up and resent (see Lines 17-19 in Algorithm \ref{mainprotocol}). If the packets are correctly recovered, then all packets follow the correct $\Delta t$ pattern; thus we consider them acceptable and not conflicting with the validity defined earlier.

By this design, one can see that the system is invulnerable to the following two common attacks: 1) an adversary commits a jamming attack and therefore the 4G module cannot send to or receive from the outer world any information; and 2) the 4G signal is generally weak or the link is disconnected. Both cases are common in practice, and appear to have the device disconnected causing the records to be missed for a period of time. Note that even when the transmission signal is down, the secure zone is still constantly monitoring the environment inside the box and this backup queue cannot be modified as it stays inside the secure zone. Also note that the adversary may compromise the inner environment of the TEE system while jamming the communications. Nevertheless, the records still stay in the secure zone (inside the backup queue), though jamming prevents them from being sent out. Unless the adversary jams the system forever or physically destroys the system, this attack action will sooner or later be publicly revealed.

Nevertheless, the system may still be vulnerable under the following scenario in which  an adversary jams the system for less than $\max\_F_d$ number of recording times and meanwhile 1) successfully breaks into the system and modify the program within TEE, or 2) forges a signature without the secret key and impersonates this device to upload fake data, or 3) steals the secret key from side channel attacks such as power analysis, then creates a valid signature. The first two cases can be nullified by the following security assumptions: 1) a trusted hardware (TEE) protects its programs in the secure zone from modifications and 2) the unforgeability proof of a secure digital signature scheme states that no polynomial time adversary can forge a valid signature with a negligible probability. The last case is thwarted by prohibiting any long disconnection time in order to make the attack window impossible to exist. Such a case is further explained in the following subsection.


\subsection{State-of-the-Art of TEE Security}
\label{sec:secTEE}
In this paper, we ensure our \textit{truthful} requirement based on the assumption that programs and systems within TEE hardware is tamper-proof. We now compare the mainstream choices of TEE platforms and list out state-of-the-art attacks and defense mechanisms. We finally show that there does not exist an effective attack against our Cortex-M series TEE-based scheme that can be avoid detection.

The Cortex-A series microcontrollers and Intel SGX platforms are the most prevalent TEE choices in smartphones and servers, and they do not strictly enforce physical isolation between the secure zone and the non-secure zone. Cortex-A macro-schedules the memory space using virtual memory MMU at the software level, which is more vulnerable to long-range attacks. Intel SGX has a Model Specific Register (MSR) that is exposed to software control, which can be modified in a controlled way to flip bits in the memory, causing the attacker to gain escalated control privilege \cite{CVE-2019-11157}.  This led to multiple successful software-level exploits \cite{murdock2020plundervolt} \cite{kenjar2020v0ltpwn}, and even attacks after the vulnerability was patched \cite{chenvoltpillager}. Cortex-A series was also exploited since it has a shared memory cache \cite{spreitzer2013cache}\cite{reinbrecht2016side}.

On the other hand, a Cortex-M series microcontroller (MCU) explicitly micro-manages the memory space, which is considered a much more solid physical isolation. By the time of this writing, no software-level exploit against the Cortex-M series has been identified. At the hardware-level, Vafa \textit{et al.} launched a power analysis-based profiling attack to uncover the code within the Cortex-M3 core~\cite{vafa2020efficient}; and Petrvalsky \textit{et al.} employed differential power analysis (DPA) and FastDTW trace alignment to recover the secret 128-bit key from the Arm Cortex-M3 MCU~\cite{petrvalsky2014differential}. Nevertheless, hardware-level sabotage and side channel attacks are technically unavoidable and are beyond our consideration. Furthermore, these attacks are mainly against secrecy -- uncovering secret keys or codes, which are non-critical to integrity, as the integrity of the TEE system remains uncompromised and its functionality is not altered. But the recovery of the secret keys~\cite{petrvalsky2014differential} could seriously threaten the system trustworthiness. Nevertheless, this attack can be easily defendable via memory masking \cite{yoshikawa2011efficient}. Aside of that, Petrvalsky \textit{et al.} in~\cite{petrvalsky2014differential} did not report the total time needed for succeeding this attack. According to the descriptions in their experimental studies, during one attack, one trace alignment using FastDTW is processed, and each trace consists of about 150,000 samples. By looking up the time cost from the original paper of FastDTW~\cite{salvador2007toward}, one can calculate that the time cost for each single trace alignment could be somewhere between \textbf{100 to 300 seconds}, which is considered as the time cost for a single attack. We call this time duration the \textit{attack window}, which means that if the total duration of 4G disconnection for our environment monitoring system is greater than this attack window, the attack case 3) mentioned in the previous subsection could succeed. Therefore, we mandate an allowed disconnection time, denoted by the $\max\_F_d$ number of unsent packets, must be less than the attack window, and the existence of a disconnection time exceeding the attack window  implies that the device could be compromised and is distrustful. 
In reality, one can adjust $\max\_F_d$ by considering different communication methods (such as delay tolerant network, 2G networks, etc), different applications, and new attacks in future. To our best knowledge, we are the first to develop secure blockchain components using the Cortex-M series chipsets, and mandate high security standards to thwart all the existing major attack schemes targeting Cortex-M series.

\begin{figure}[!t]
	\centering
	\includegraphics[width=0.49\textwidth]{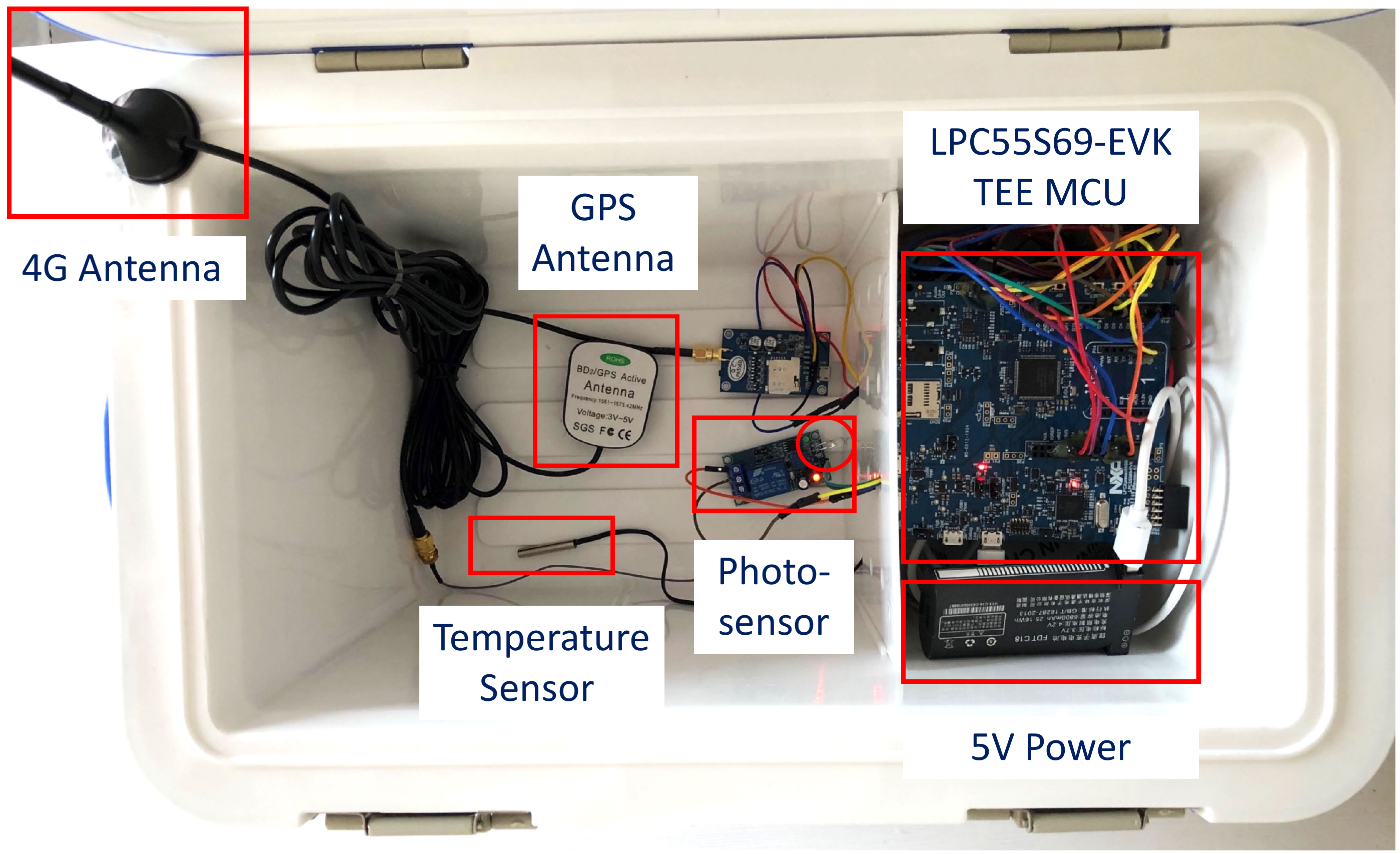}
	\caption{The vaccine shipping box with sensing devices}
	\label{fig:photo}
\end{figure}

\section{Experiments}
\label{sec:exp}
In this section, we put our system into an actual test. We implemented the system as described earlier, attached it into a vaccine shipping box, and performed consistency monitoring towards the vaccine transportation.

\begin{figure*}[!ht]%
	\centering
	\subfloat[GPS movement trace and security checkpoints] {{\includegraphics[width=0.3\textwidth]{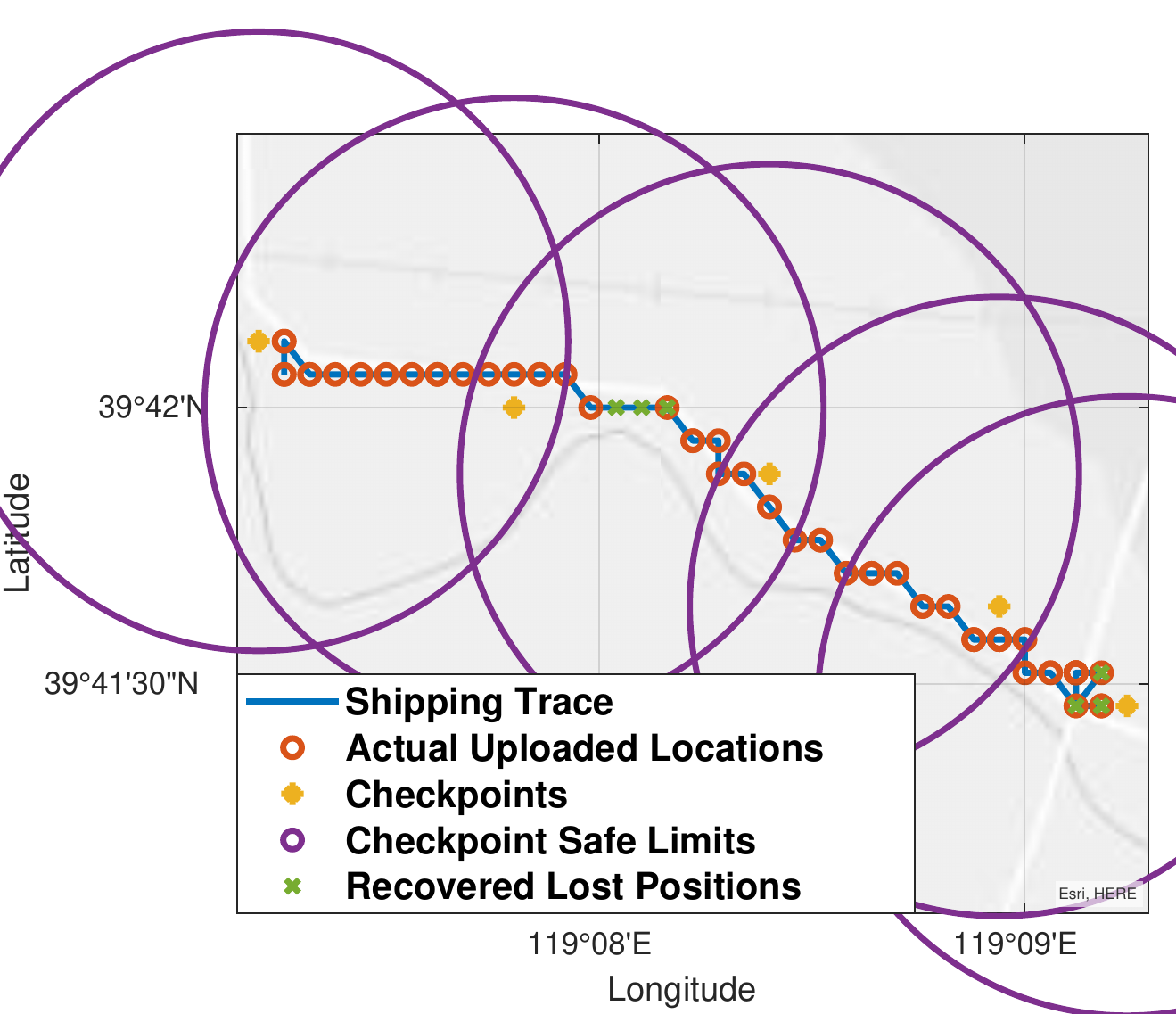} \label{fig:GPStrace}}}%
\qquad
	\subfloat[Temperature fluctuation] {{\includegraphics[width=0.3\textwidth]{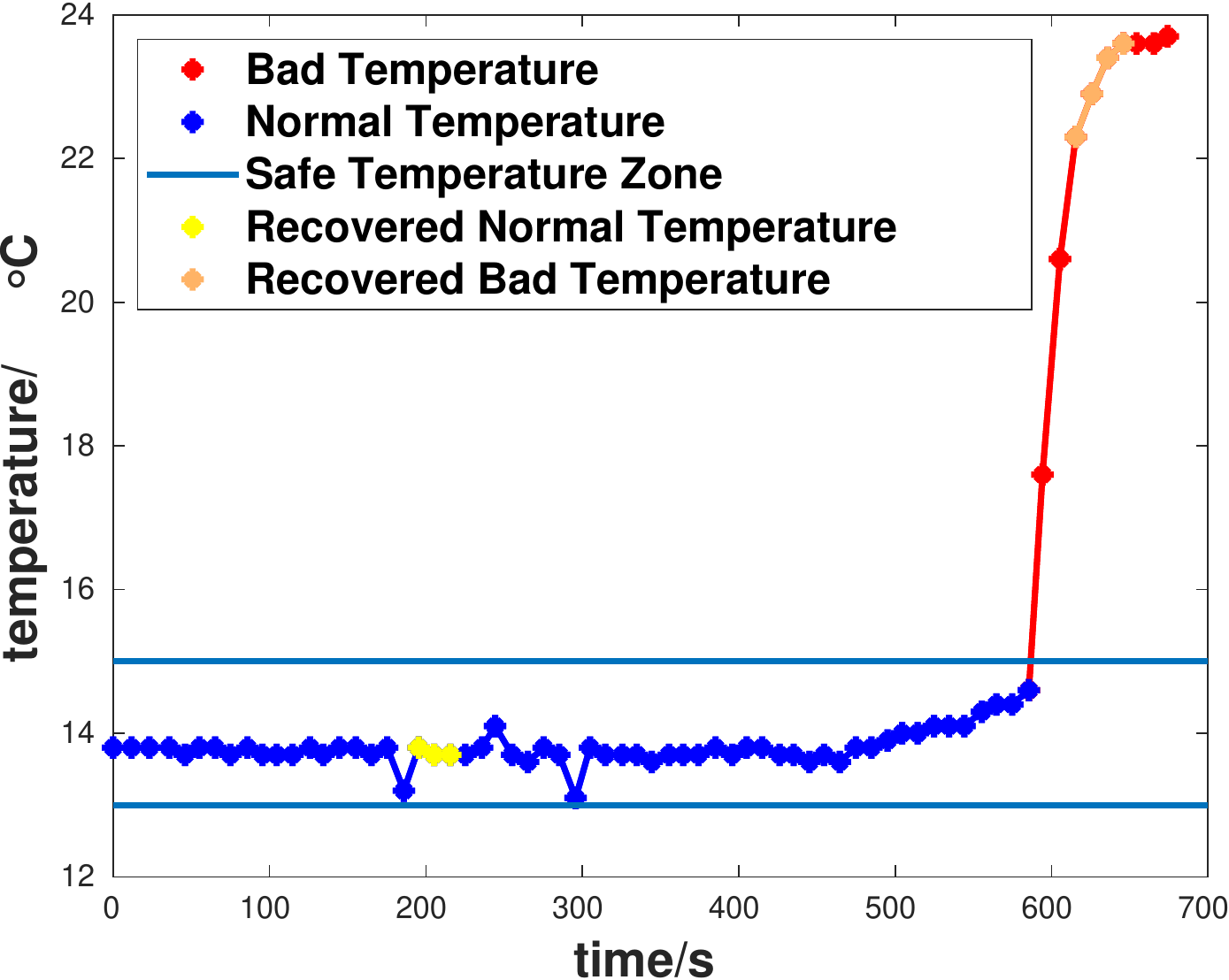}
	\label{fig:temp}}}%
\qquad
	\subfloat[Photosensor and its logic voltage level output] {{\includegraphics[width=0.3\textwidth]{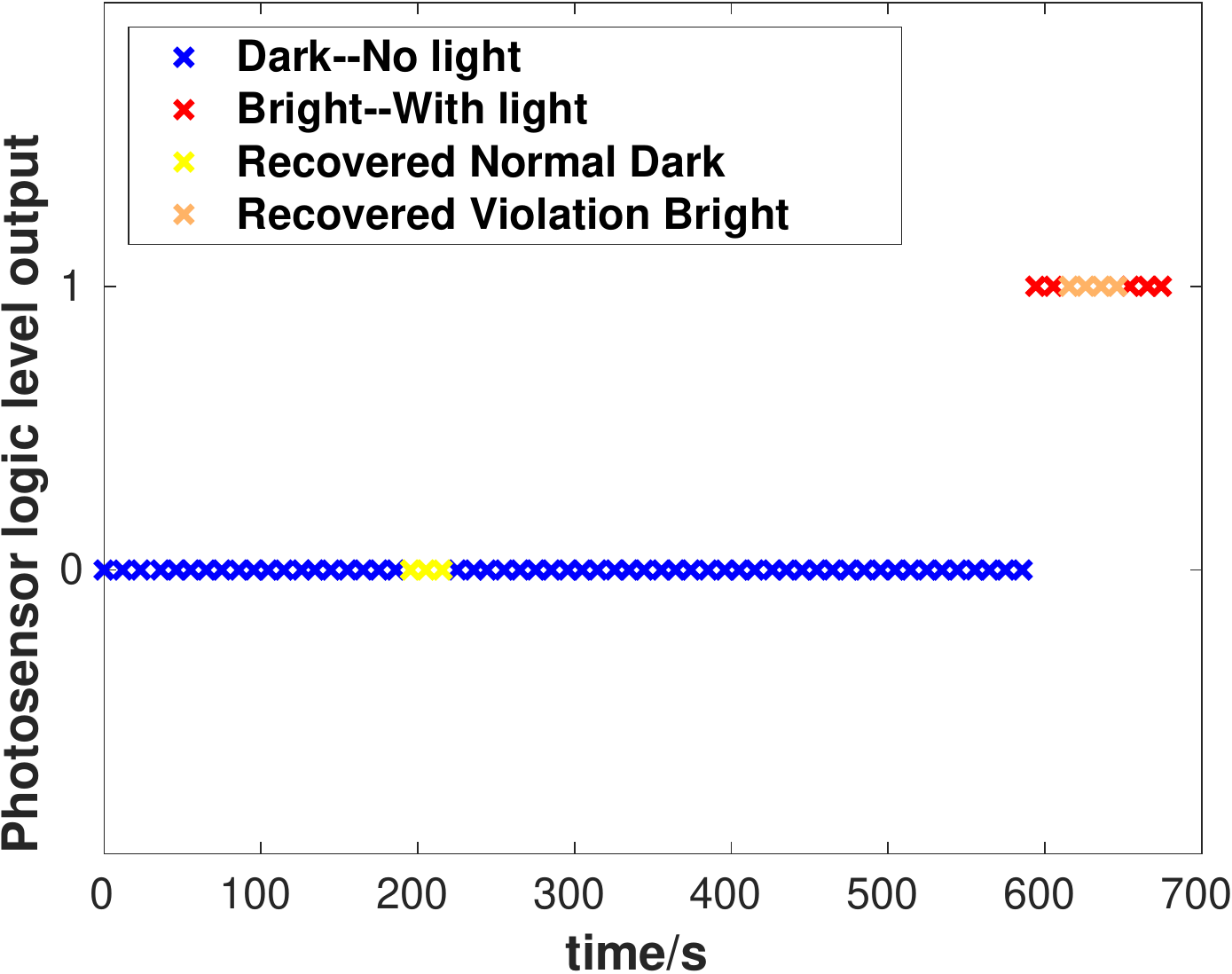}
	\label{fig:binarylux}}}%
\qquad
	\caption{Experimental Sensing Data Separately Visualized as GPS Location, Temperature and Brightness. }%
\end{figure*}

\subsection{Setup}
The system was implemented on an LPC55S69-EVK development board from the NXP Semiconductor. The board consists of a LPC55S69 dual-core Arm Cortex-M33 microcontroller, running at 150MHz, and supports Arm TrustZone technology. For more information regarding LPC55S69-EVK please refer to the user's manual \cite{LPC55S6xMCU}

We developed our own blockchain system using Golang \cite{golang} for the best flexibility support. Golang is quite popular in security community for its memory-safe, high-concurrent, and high-usable properties. We adopted the Tendermint-BFT consensus algorithm, which is the Delegated Proof-of-Stake version of the original Practical Byzantine Fault Tolerance mechanism \cite{buchman2018latest}. Specifically, Tendermint-BFT assigns different weights to different nodes during BFT voting, while PBFT assigns equal weights to all nodes. This weight, in practice, can be used to represent different trustworthiness of different nodes. We implemented the main blockchain system with about 4000 line of codes in Golang. The PC we used is a 8-Core Intel i7-6700HQ @ 2.6GHz with 16G memory and Ubuntu 18.04.1 GNU/Linux. We simulated 6 blockchain nodes on top of this PC as our blockchain system. The vaccine shipping box communicates with the blockchain through the civil-use commercial level 4G-LTE, which has an actual bandwidth around 10 Mbps.

\subsection{Evaluation}

In our evaluation, we burned the project (containing codes for both the secure zone and the non-secure zone) into a LPC55S69-EVK board system. By providing a 5V external power bank, we made the system a standalone one.  The board was positioned to the right side of the inside vaccine shipping box, leaving all critical sensors on the left side. As one can see from Fig. \ref{fig:photo}, the temperature sensor was the silver textured probe, and the main part of the photosensor was a light-sensitive LED. The GPS antenna was small in size and its signal power was not affected when sealed inside the box. The 4G antenna  was attached to the exterior top-left of the box. We set $\Delta t$ to be 10 seconds, which means that we sensed the environment and uploaded to the blockchain one record every 10 seconds. The maximum number of blocks for error recovery was fixed to 5. In our experiment, we attached the vaccine shipping box at the back seat of a motorcycle. The safe vaccine transportation temperature was ranged from 13 to 15 degrees.

We rode the motorcycle along a street in a suburban area and continuously collected the sensor data. To decrease the inside temperature of the vaccine box we placed two ice bags. The following results represent the data collected from 19:15:19 to 19:26:47 after the temperature dropped to the safe range, during which the 4G antenna was unplugged from 19:18:48 to 19:19:08 and from 19:25:48 to 19:26:18 to mimic two short jamming attacks. The vaccine box was opened at 19:25:27 and remained open until 19:26:47. As the sensing interval is 10 seconds, we collected in total 92 data points. Due to jamming we missed 3 records from the first attack and 4 records from the second one.

%

From the 92 data points one can experimentally investigate the critical latencies of the system, i.e., $\epsilon_1, \epsilon_2$, and $\epsilon_3$. 
We observed that  $\epsilon_1$ was between 50 to 158 ms with a mean of 84 ms and a variance of 0.002 $ms^2$. 
When the 4G signal was stable, the 4G transmission latency $\epsilon_2$ was between 50 and 80 ms, with a mean of 76.3 ms and a variance of $2.2 \times 10^{-5}$ $ms^2$.  Each data record was about tens of KB so it did not inflict huge transmission overhead. The blockchain latency $\epsilon_3$ was less than 60 ms with a mean of 53.85 ms and a variance of $2.36 \times 10^{-7}$ $ms^2$, as we had only 6 blockchain nodes. The maximum latency among all the data was 7.78 seconds. 
So one can safely assert that the total latency upper limit was 8 seconds for our test system.

Fig. ~\ref{fig:GPStrace} visualizes the actual shipping trace. We setup five predefined security checkpoints (gold stars) and fixed a radius of 500 meters (purple big circles). The security checkpoints and radius together define an approved ground transportation trace, which is a road shown as a faint white trace from left to right \footnote{Not to confuse with the grey one, which is a river.}. As one can see from the figure, the location was recorded, uploaded and authenticated regularly (red dots), demonstrating a trusted actual transportation trace. Although there were 3+4 planned jamming losses (green crosses), the records were all resent and recovered by the next following one. The 3 dropped records in the middle were clear to see as green crosses, with one of them overlapping with another data. The 4 in the end were all overlapping with others as we were circling in a small area. As one can see from the figure, the driver clearly did not deviate from the predefined route.


During transportation,  we first kept the temperature low and made it stay steady at 14 degrees Celsius. Fig.~\ref{fig:temp} shows the recordings of our actions, where we kept the box closed and sealed during time 0-600 (blue stars), and then opened it (at a summertime) after time 600 (red stars). Note that the $x$-axis uses the relative time from 0 to 700 instead of the actual data collection time for better illustration. After the box was opened, its internal temperature was quickly increased for 10 consecutive records and stayed stable at around 24 degrees. When these records were synchronized at all blockchain nodes, it was easy to observe and prove that this vaccine was exposed to undesirable temperature and should be no longer acceptable for customer use. Our planned jamming losses at time 200 (yellow stars) and 600 (orange stars) were also recovered, as shown in this figure.


The outputs of the photosensor were illustrated in Fig. \ref{fig:binarylux}. One can see clearly that the sealed vaccine box was opened at time 600. When the box was sealed, it remained dark (with a proper threshold that can negate faint LEDs inside the box) and hence the photosensor constantly output 0 (blue crosses); while when the box was opened at time 600 and was exposed to the sun, the photosensor snapped the ray and output the 3.3 V high logic voltage level to the board (red crosses indicating the logical output of 1 in Fig. \ref{fig:binarylux}). The photo and temperature sensors together can enhance the confidence that the box and its inside were indeed compromised. Fig. \ref{fig:binarylux} also demonstrates that our 3+4 planned jamming losses at around time 200 and 600 were recovered correctly.


\section{Discussions}
\label{sec:secfuture}


We emphasize that this paper presents a general framework that extends the on-chain trust to off-chain in a truthful, real-time and continuous manner. It is not confined to the vaccine tracing use case -- the framework can be extended to any use case that meets the requirements specified in Section~\ref{sec:def}. Most importantly, future implementations must follow Definition~\ref{def:trust}, which defines the three properties of trustworthiness, i.e., truthfulness, real-time consistency and continuity, to ensure the trustworthy mapping between a physical object and its digital entity.


Our framework scheme might suffer from the following two security limitations. 1) As we have claimed in the trust model section, physical damages towards TEE is out of scope of this paper. Such type of attacks is actually very possible. For example, since our method retrieves data using the \texttt{Sensors\_get} interface from sensors every $\Delta t$ time, an attacker may physically remove the sensor and weld a malfunctioning one within the $\Delta t$ window. 2) The monitoring sensors can only capture the events occurring in the close neighborhood of the physical object but cannot prove the authenticity of the object. Therefore we cannot be 100\% sure whether or not the object is replaced by a fake midway when the sensors are disabled or when the object is out of their sensing ranges -- we can only \textit{suspect that the object could have been mistreated when the sensors are disabled}. A typical example is the wine forgery attack which may happen during transportation, where an attacker may first disable all sensors, then drill a tiny hole on the wine body, extract all the real wine, replace it with fake one, and finally plug the hole to finish up. 


Both limitations mentioned above can be mitigated using Physical Uncloneable Function (PUF). Limitation 1 is essentially an authentication problem -- in our current design we do not authenticate our sensors, and any connection through wire is considered secure. PUF-enabled sensors can achieve efficient authentication and thus provide a practical security improvement, ensuring not only our code programmed in TEE is tamper-proof, but also our hardware component connected to TEE is tamper-proof. Limitation 2 is a fundamental problem in product anti-forgery tracing, and digitally verifying the authenticity of a physical object is notoriously hard. One possible solution is to design a \textit{digital seal} using PUF. A digital seal is like a regular seal, but can electronically detect sabotage by continuously verifying the integrity of the protective physical material. In the scenario of wine anti-forgery tracing, one can cover the wine bottle with a double-layer thin metal membrane, which works as a capacitance. A PUF-enabled detector is then connected to the positive and negative side of the capacitance to continuously monitor the capacitance value in each pixel, then output the values in the structure of a matrix. As any touching or drilling can significantly affect the capacitance values of an area, the mis-behaviors can then be identified. In our future research, we will explore PUF based solution approaches to overcome the limitations mentioned above.  

\section{Related Work}
\label{sec:related}
There exist a few attempts to extend the blockchain trustworthiness from on-chain to the physical world, by either applying blockchain directly to the physical scenarios making use of blockchain security, or using secure hardware or cryptographic primitives to enhance blockchain's internal security and trustworthiness.

Blockchain technologies have been applied in the fields of smart city, autonomous driving, smart home, etc. Zhou \textit{et al.} proposed a hierarchical IoT architecture for smart home environments, where each home maintains a private chain that can hierarchically construct into a public one. However there were no countermeasures or discussions regarding how to stop the home owner from faking the record and lying to the public chain and how to validate the subchain data authenticity at the public chain. Guo \textit{et al.} developed an event recording system for autonomous vehicles with blockchain, and presented the concept of ``proof of event'', which 
is a hash digest of a combination of event locations, embedded vehicle event recorder readings and timestamps. This digest does not include any unforgeable data or data from an authenticated source, thus it can be easily faked afterwards \cite{guo2018blockchain}. It is also vulnerable to sybil attacks, where the fake data is shared among all sybil identities and uploaded by all, in order to increase the confidence of the fake data.

The above examples demonstrate a common problem shared by most blockchain applications: the blockchain trustworthiness is confined within the on-chain environment. There is no strong guarantee that the on-chain records continuously and truthfully reflect the \textit{true} physical world, and that the blockchain control instructions are forcibly deployed on time.

Secure hardware or cryptographic primitives were also employed to enhance the blockchain security. Lind \textit{et al.} proposed TEEChain, a TEE-based secure payment blockchain network in which TEE was used as \textit{treasuries} to manage the off-chain funds and payments inside the TEE secure zone \cite{lind2019teechain}. Ayoade \textit{et al.} used the Intel SGX to securely offload data from on-chain to SGX-powered databases, which can hash the stored data and compare to the correct hash stored in the blockchain for data correctness verification. Dang \textit{et al.} managed to scale blockchain via sharding and Intel SGX. They made use of the protected enclave module in Intel SGX to enhance security and increase the performance of the Proof-of-Elapsed-Time consensus and the PBFT consensus algorithms. Chainlink \cite{Chainlink} and Provable \cite{Provable} adopted various secure hardware such as Intel SGX to build a trusted inter-blockchain data exchange solution or website-to-blockchain trusted data input oracle. Some of these works made impressive progress by achieving trust extension between blockchains or between a blockchain and other digital environments, but they cannot extend trust from on-chain to the physical world. Not to mention that Intel SGX is a server-level hardware that typically costs around \$300 to \$400, which highly restricts its adoptions by edge applications.

In this paper we propose a scheme to extend trust from on-chain to the off-chain physical world by developing a TEE-enabled trusted environment monitoring system and a fault tolerance uploading protocol to achieve real-time high data consistency between on-chain digital world and off-chain physical world. As far as we know, we are the first to overcome all engineering challenges and develop such a trustworthy system over the Cortex M33 MCU, which is highly competitive in price --  around \$40, approximately one-tenth of that of Intel SGX.

\section{Conclusions}
\label{sec:conclusion}

In this paper, we propose a scheme to extend blockchain trust from on-chain to off-chain, and implement the full system taking trustworthy vaccine tracing as an example. Our scheme consists of a TEE-enabled Trusted Environment Monitoring System that continuously senses and generates anti-forgery data, and a consistency protocol that uploads data from the system to blockchain in a truthful, real-time consistent, continuous, and fault-tolerant way. Our experiment records the internal status of the vaccine shipping box during the whole shipping process. One can see that our system is quite efficient with approximately 7-second of system processing time, 70ms of transmission time, and 40ms of blockchain synchronizing time. Our planned jamming attacks and physical violations are also captured and errors are recovered as expected.

As an exploratory work, we select vaccine tracing as an example and choose the photosensor, GPS sensor and temperature sensor to describe the status of a moving vaccine box. As we have our full system available, which is a relatively general framework that can customize critical parameters and tools, researchers can extend to any application of their interest by developing their own additional sensing devices based on our system template. We keep this work open-sourced at \url{https://github.com/zhuaiballl/TEE-enabled_Trusted_Environment_Monitoring_System}.

\section{Acknowledgement}
It was partially supported by the National Key R\&D Program of China under grant 2019YFB2102600, the National Natural Science Foundation of China under grants U1811463, 61971014, 61832012, 61771289, and 11675199, and the National Science Foundation of the US under Grants IIS-1741279 and CNS-1704397.

\bibliographystyle{IEEEtran}
\bibliography{liu}


\begin{IEEEbiography}[{\includegraphics[width=1in,height=1.25in,clip,keepaspectratio]{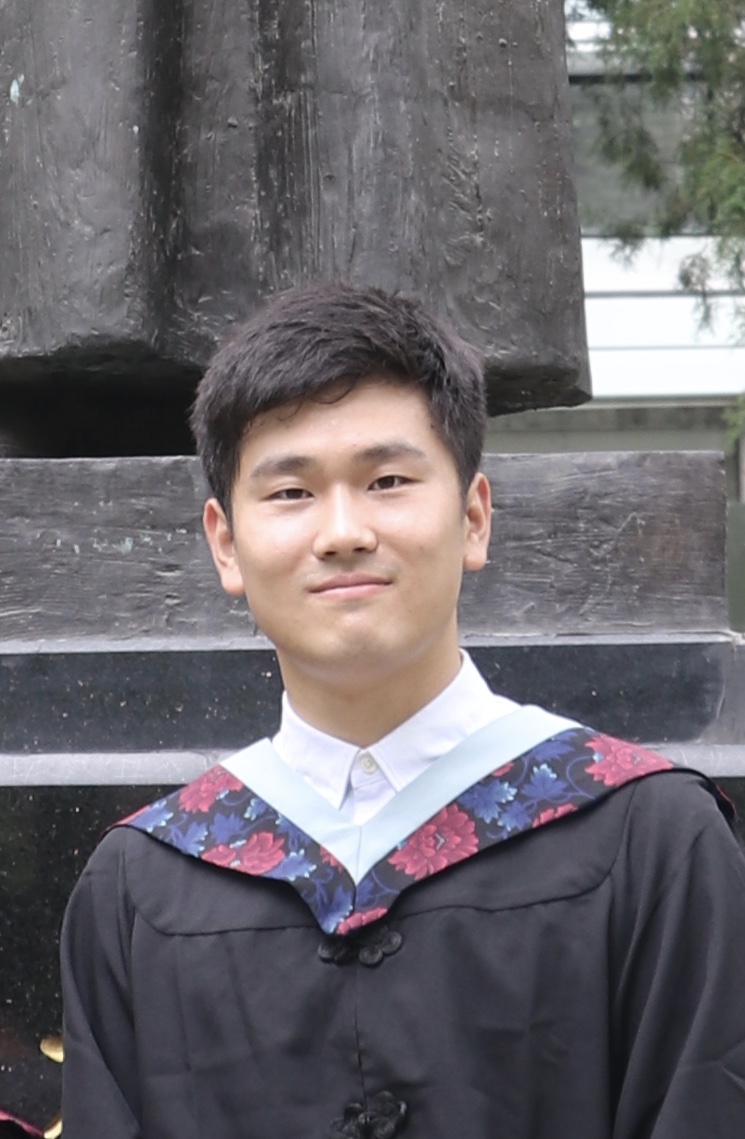}}]{Chunchi Liu} received his PhD degree from The George Washington University, Washington DC, USA, in 2020, and his BS degree with distinction from Beijing Normal University, Beijing, China, in 2017, both in Computer Science. He is now the Lead Scientist at Ernst \& Young Greater China Area and is a member of EY global executive management division. His current research focuses on Blockchain, Internet of Things, Applied Cryptography and Security.
\end{IEEEbiography}


\begin{IEEEbiography}[{\includegraphics[width=1in,height=1.25in,clip,keepaspectratio]{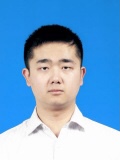}}]{Hechuan Guo} is a PhD student in Computer Science at Shandong University, Qingdao, China. He received his BS Degree in Computer Science in 2017 and MS degree in Engineering in 2020, both from Beijing Normal University, Beijing, China. His current research focuses on Blockchain, Consensus Protocols, Security, and Applied Cryptography. 
\end{IEEEbiography}

\begin{IEEEbiography}[{\includegraphics[width=1in,height=1.25in,clip,keepaspectratio]{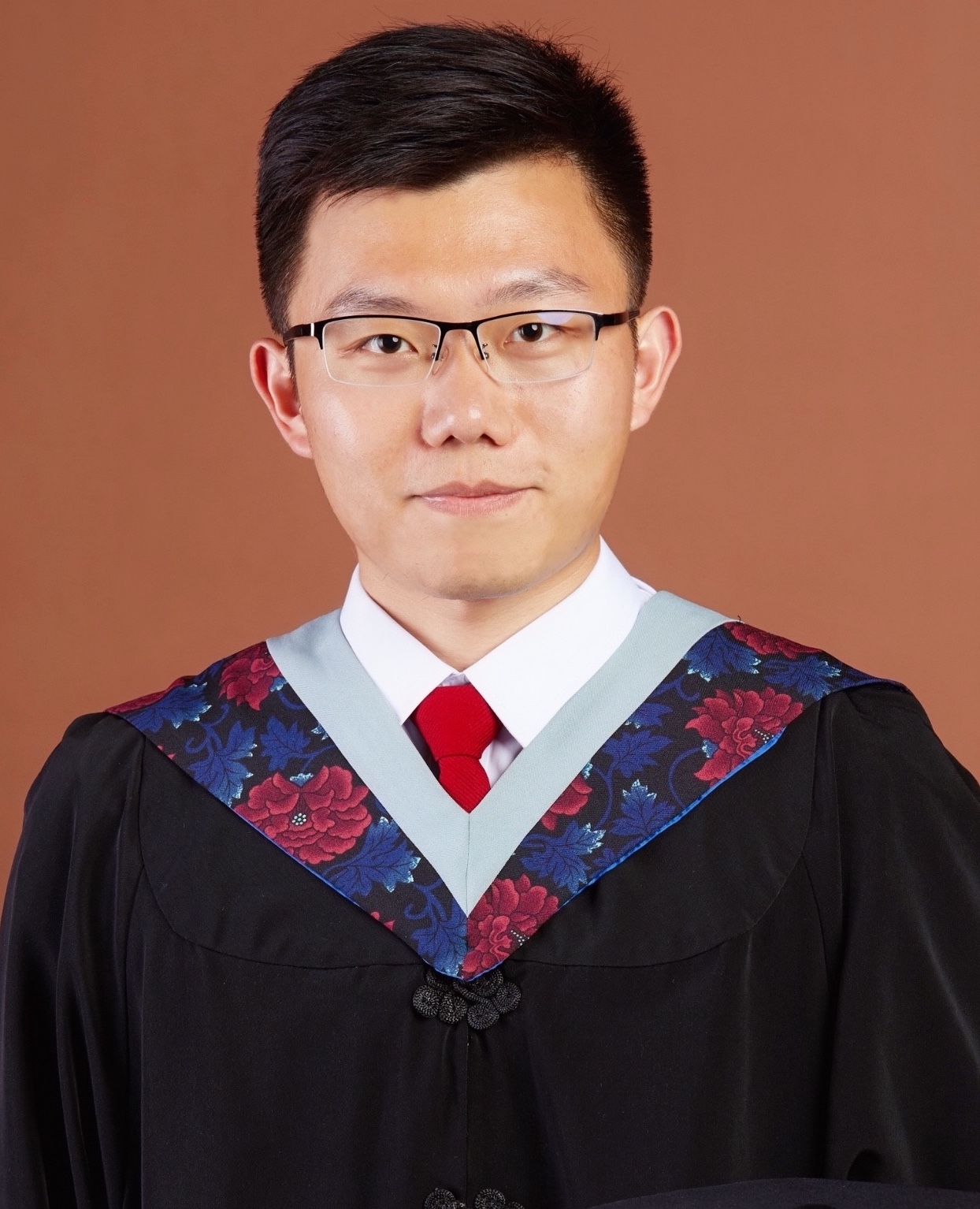}}]{Minghui Xu} is a PhD student in Computer Science at The George Washington University, Washington DC, USA. He received his BS degree in Physics in 2018 and minored in Computer Science during 2016-2018 from Beijing Normal University, Beijing, China. His current research focuses on distributed computing, blockchain, and quantum computing.
\end{IEEEbiography}

\begin{IEEEbiography}[{\includegraphics[width=1in,height=1.25in,clip,keepaspectratio]{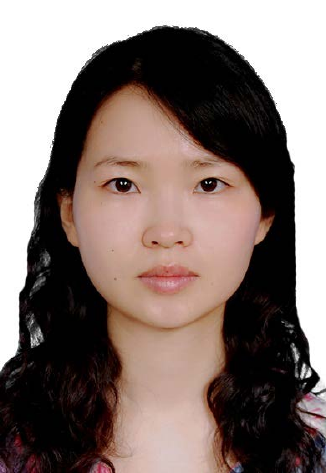}}]{Shengling Wang} is a professor in College of Information Science  and Technology, Beijing Normal University. She received her Ph.D. in 2008 from Xi’an Jiaotong University. Then she did her postdoctoral research in the Department of Computer Science and Technology at Tsinghua University from 2008 to 2010, and worked as a faculty member from 2010 to 2013 in the Institute of Computing Technology of CAS. Her research focuses on mobile/wireless networks, game theory, and crowdsourcing.
\end{IEEEbiography}

\begin{IEEEbiography}[{\includegraphics[width=1in,height=1.25in,clip,keepaspectratio]{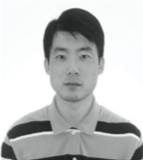}}]{Dongxiao Yu} received his BS degree in Mathematics in 2006 from Shandong University, and PhD degree in Computer Science in 2014 from The University of Hong Kong. He became an associate professor in the School of Computer Science and Technology, Huazhong University of Science and Technology, in 2016. Currently he is a professor at the School of Computer Science and Technology, Shandong University. His research interests include wireless networking, distributed computing, and graph algorithms.
\end{IEEEbiography}

\begin{IEEEbiography}[{\includegraphics[width=1in,height=1.25in,clip,keepaspectratio]{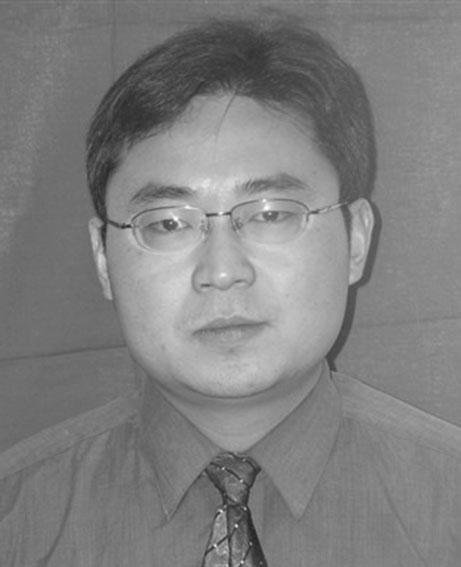}}]{Jiguo Yu} received his Ph.D. degree in School of Mathematics from Shandong University in 2004. He became a full professor in the School of Computer Science, Qufu Normal University, Shandong, China in 2007. Currently he is a full professor in Qilu University of Technology (Shandong Academy of Sciences), His main research interests include privacy-aware computing, wireless networking, distributed algorithms, blockchain, and graph theory.
\end{IEEEbiography}

\begin{IEEEbiography}[{\includegraphics[width=1in,height=1.25in,clip,keepaspectratio]{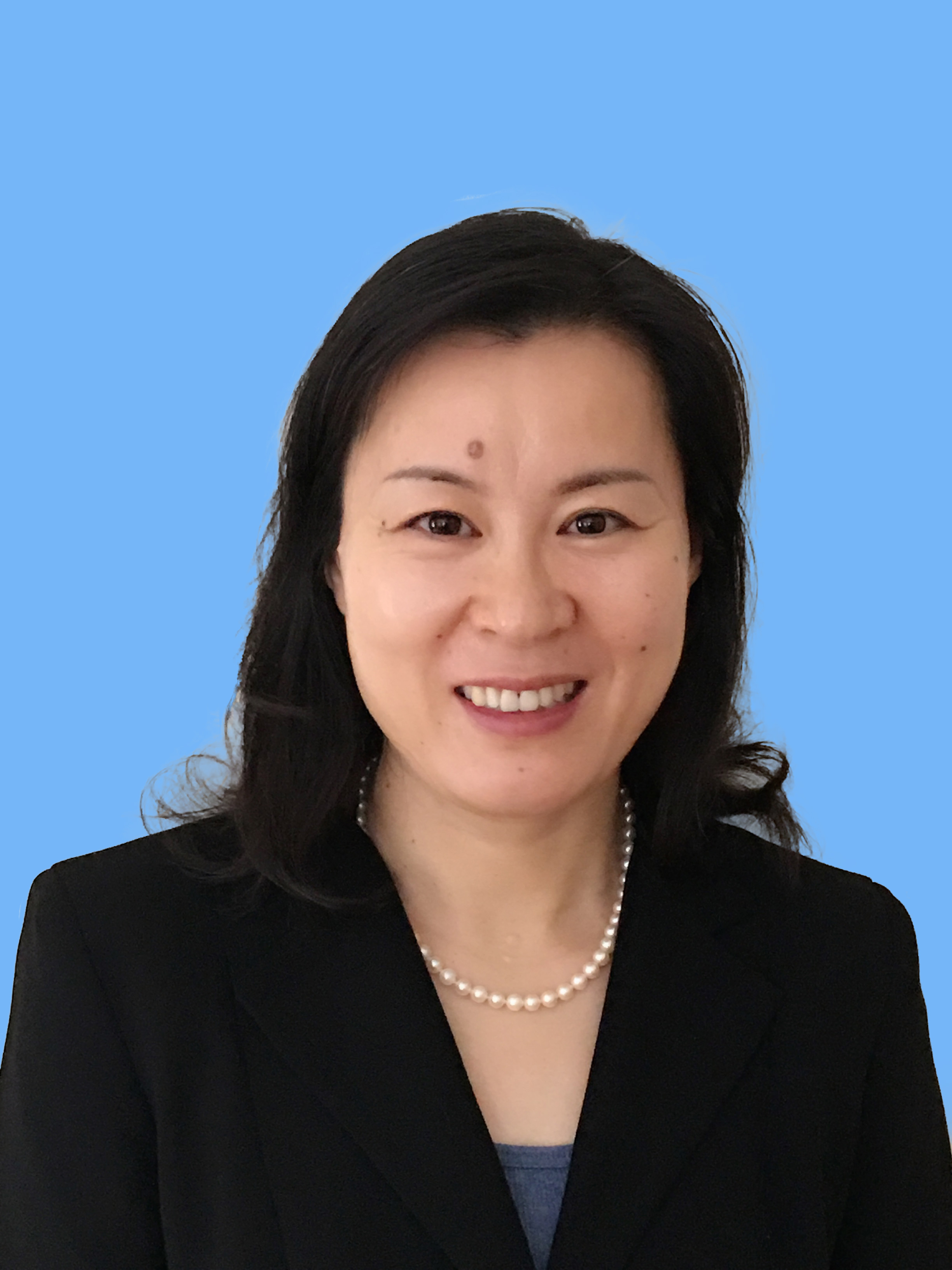}}]{Xiuzhen Cheng} received her MS and PhD degrees in computer science from University of Minnesota, Twin Cities, in 2000 and 2002, respectively. She was a faculty member at the Department of Computer Science, The George Washington University,  from 2002-2020. Currently she is a professor of computer science at Shandong University, Qingdao, China. Her research focuses on blockchain computing, security and privacy. She is a Fellow of IEEE.
\end{IEEEbiography}

\end{document}